\documentclass[
    a4paper,
    USenglish,
    cleveref,
    autoref,
    thm-restate,
]{lipics-v2021}

\nolinenumbers
\hideLIPIcs

\usepackage{algorithm}
\usepackage{algpseudocode}
\usepackage{cases}
\usepackage{mathtools}
\usepackage{amsmath}
\usepackage[multiple]{footmisc}

\usepackage{tikz}
\usetikzlibrary{backgrounds}
\makeatletter

\tikzset{%
  fancy quotes/.style={
    text width=\fq@width pt,
    align=justify,
    inner sep=1em,
    anchor=north west,
    minimum width=\linewidth,
  },
  fancy quotes width/.initial={.5\linewidth},
  fancy quotes marks/.style={
    scale=5,
    text=white,
    inner sep=0pt,
  },
  fancy quotes opening/.style={
    fancy quotes marks,
  },
  fancy quotes closing/.style={
    fancy quotes marks,
  },
  fancy quotes background/.style={
    show background rectangle,
    inner frame xsep=0pt,
    background rectangle/.style={
      fill=gray!8,
      rounded corners,
    },
  }
}

\newenvironment{fancyquotes}[1][]{%
\noindent
\tikzpicture[fancy quotes background]
\node[fancy quotes opening,anchor=north west] (fq@ul) at (0,0) {``};
\tikz@scan@one@point\pgfutil@firstofone(fq@ul.east)
\pgfmathsetmacro{\fq@width}{\linewidth - 2*\pgf@x}
\node[fancy quotes,#1] (fq@txt) at (fq@ul.north west) \bgroup}
{\egroup;
\node[overlay,fancy quotes closing,anchor=east] at (fq@txt.south east) {''};
\endtikzpicture}

\makeatother

\usepackage{subcaption}
\usepackage{tabularx}
\usepackage{booktabs} 

\usepackage[textsize=tiny]{todonotes}

\title{TSP Escapes the $O(2^n n^2)$ Curse} 

\titlerunning{TSP Escapes the $O(2^n n^2)$ Curse} 

\author{Mihail Stoian}{University of Technology Nuremberg, Germany\and \url{https://stoianmihail.github.io}}{mihail.stoian@utn.de}{https://orcid.org/0000-0002-8843-3374}{}

\authorrunning{M. Stoian} 

\Copyright{Mihail Stoian} 



\category{} 

\relatedversion{} 

\usepackage{amsthm}
\newcommand{\sparagraph}[1]{\vspace{1mm}\noindent {\bf #1}}

\newcommand{\dprec}{\mathrm{dp}}
\makeatletter
\newcommand{\ovee}{\mathbin{\mathpalette\make@circled\vee}}
\newcommand{\make@circled}[2]{%
  \ooalign{$\m@th#1\smallbigcirc{#1}$\cr\hidewidth$\m@th#1#2$\hidewidth\cr}%
}
\newcommand{\smallbigcirc}[1]{%
  \vcenter{\hbox{\scalebox{0.82}{$\m@th#1\bigcirc$}}}%
}

\makeatletter
\newcommand{\subalign}[1]{%
  \vcenter{%
    \Let@ \restore@math@cr \default@tag
    \baselineskip\fontdimen10 \scriptfont\tw@
    \advance\baselineskip\fontdimen12 \scriptfont\tw@
    \lineskip\thr@@\fontdimen8 \scriptfont\thr@@
    \lineskiplimit\lineskip
    \ialign{\hfil$\m@th\scriptstyle##$&$\m@th\scriptstyle{}##$\hfil\crcr
      #1\crcr
    }%
  }%
}
\makeatother

\begin{document}

\maketitle

\begin{abstract}
The dynamic programming solution to the traveling salesman problem due to Bellman, and independently Held and Karp, runs in time $O(2^n n^2)$, with no improvement in the last sixty years. We break this barrier for the first time by designing an algorithm that runs in deterministic time $2^n n^2 / 2^{\Omega(\sqrt{\log n})}$. We achieve this by strategically remodeling the dynamic programming recursion as a min-plus matrix product, for which faster-than-na\"ive algorithms exist.
\end{abstract}

\section{Introduction}\label{sec:introduction}

The Traveling Salesman Problem (TSP) needs no introduction: It is probably one of the most studied problems in computer science since its establishment itself as a field (and not only; we recommend Schrijver's~\cite{schrijver_history} survey). Due to its simplicity, namely finding the shortest tour through $n$ cities, it is also one of the introductory problems in dynamic programming. Bellman~\cite{bellman_dp} and, in an independent work Held and Karp~\cite{held_karp_dp}, proposed a dynamic programming solution that runs in $O(2^n n^2)$-time and takes $O(2^n n)$-space, using a table indexed by vertex subsets. This has been \emph{the} textbook solution for decades.

The unweighted case, namely deciding whether the graph has a Hamiltonian cycle, can indeed be solved faster: In his breakthrough, Bj\"orklund~\cite{hamilton_determinant} showed a Monte Carlo algorithm that runs in $O^*(1.657^n)$-time. Using a standard technique, it can be adapted to solve our particular problem, yet with an $O(w)$ dependence in the running time, where $w$ is the sum of all edge weights. In this context, Nederlof~\cite{nederlof_tsp_bipartite} provided an $O(1.9999^n)$-time randomized algorithm with constant error probability for the \emph{bipartite} TSP, \emph{assuming the matrix-multiplication exponent, $\omega$, is equal to 2}. He also outlines how his methods can be useful for solving TSP in non-bipartite graphs. Degree-boundedness is also a fertile ground for such improvements~\cite{tsp_bounded_1, tsp_bounded_2}. In the Euclidean setting, the $2^{O(1 - 1/d)}$ lower-bound~\cite{euclid_tsp_lb} has already been attained~\cite{euclid_tsp_algo}. However, despite these advances, the standard dynamic programming solution in time $O(2^n n^2)$ is still the best algorithm for the general setting. This was also the status quo when Cook~\cite{cook_tsp} wrote his compendium on TSP ten years ago:
\newline

\begin{fancyquotes}
\centering
Unfortunately for TSP fans, no good algorithm is known for the problem. The best result thus far is a solution method, discovered in 1962, that runs in time proportional to $n^2 2^n$.\\
--- William Cook~\cite{cook_tsp}
\end{fancyquotes}

This leads us to our driving research question:
\begin{quote}
    \itshape
    \centering
    Is there an $o(2^n n^2)$-time algorithm for TSP?
\end{quote}
A positive answer would break another \emph{psychological} barrier, as Koutis and Williams~\cite{koutis_williams} refer to it, and, if simple enough, can be used as the new textbook algorithm for TSP.

\section{Our Results}

We answer our research question in the positive. We introduce an algorithm that, to our best knowledge, is the first to break the longstanding $O(2^n n^2)$ time-barrier for TSP. We remodel the standard dynamic programming recursion as a min-plus matrix product, namely:
\begin{restatable}[Backbone]{theorem}{MainResult}
\label{thm:weakly_approx_min_sum_subset_convolution}
    If the Min-Plus Matrix Product of $n \times n$ matrices can be solved in time $T(n)$, then TSP can be solved in time $O(2^n T(n) / n)$.
\end{restatable}

If we plug in the fastest algorithm for the min-plus matrix product due to Williams~\cite{williams_apsp}, which was been later derandomized by Chan and Williams~\cite{chan_williams}, we obtain our claimed result:

\begin{corollary}[Main]
    TSP can be solved in deterministic time $2^n n^2 / 2^{\Omega(\sqrt{\log n})}$.
\end{corollary}

The factor $\frac{1}{2^{\Omega(\sqrt{\log n})}}$ may seem a bit unusual to the untrained reader. In particular, this factor is better than $\frac{1}{\log^k n}$ time, for every $k$. We next outline our algorithm.

\section{Algorithm}

Let us first recall the standard dynamic programming (DP) solution:
\begin{align}
    \dprec(S, k) &= \displaystyle\min_{j \in S \setminus \{k\}}\left(\dprec(S \setminus \{k\}, j) + c_{j,k}\right), k \in S\label{eq:rec}\\
    \dprec(\{1\}, 1) &= 0\label{eq:base_case},
\end{align}
where $S \subseteq [n] \vcentcolon= \{1, \ldots, n\}$ and $c$ is the $n\times n$ cost matrix of the instance. At the end, the final solution is computed as $\displaystyle\min_{j}\left(\dprec([n], j) + c_{j,1}\right)$. The algorithm runs in time $O(2^n n^2)$, since for each state $(S, k)$ we have to loop over all $j$'s in $S \setminus \{k\}$. A convenient way to compute the $\dprec$-table is to fix a cardinality $\ell \in \{2, \ldots, n\}$ and compute the $\dprec$-values for all sets of cardinality $\ell$. In other words, the $\dprec$-table is computed \emph{layer-wise}.

\sparagraph{Key Insight.}~The next observation is rather crucial for understanding our algorithm: To some extent, Eq.~\eqref{eq:rec} \emph{is indeed an inner product in the min-plus semi-ring} between the vector $\dprec(S \setminus \{k\}, \vcentcolon)$ and the $j$th column of the cost matrix $c$ (for a fixed $k \in S$). Thus, the computation of the entire row $\dprec(S, \vcentcolon)$ consists of several min-plus inner products. Next, we exploit this new perspective.

\sparagraph{Remodeling the DP.}~To this end, we remark that if we represented these inner products as a matrix product, the output could also be shared with other rows of the current layer. With this in mind, we can \emph{group} the rows of the $\dprec$-table of the previous layer, i.e., the $(\ell - 1)$th layer, into batches of size $n$ (so that the dimension matches).\footnote{The rows of the final batch may be padded with $\infty$ in case the batch does not have exactly size $n$.} We perform the min-plus product of the current batch with the cost matrix and then update the rows for the current cardinality $\ell$. After all $\lceil{n \choose \ell - 1} / n\rceil$ batches have been iterated, the $\ell$th layer of the $\dprec$-table has been fully computed.

\sparagraph{Pseudocode.}~We outline the pseudocode in Alg.~\ref{alg:new_dp}. We first initialize the table with $\infty$ and set up the base case, Eq.~\eqref{eq:base_case}. We then iterate over all cardinalities $\ell \in \{2, \ldots, n\}$ and generate the corresponding batches. Specifically, the function $\textsc{Batches}(\ell - 1)$ generates internally the set $\{S \subseteq [n] \mid |S| = \ell - 1\}$ containing all sets of cardinality $\ell - 1$ and then chunks it into batches of size $n$ (the final batch is handled individually in the min-plus~product), returning these one by one in the variable $\mathcal{B}$. For each batch, we perform the min-plus product between the corresponding rows of the $\dprec$-table, i.e., $\dprec(\mathcal{B}, \vcentcolon)$, and the cost matrix $c$. This results in a temporary matrix $p$ which is then iterated to fill in the values of the current layer. Namely, for each set $\mathcal{B}[i]$, we take the $k$'s that are not present within it and update the value of the set $\mathcal{B}[i] \cup \{k\}$ with $p_{i,k}$. In particular, $p_{i,k}$ stores the min-plus inner product between $\dprec(\mathcal{B}[i], \vcentcolon)$ and $c_{\vcentcolon, k}$. Hence, the operation in Eq.~\eqref{eq:rec}, which once took linear time, can now be done in $O(1)$, given the min-plus product $p$ was computed before. Finally, we return the solution.
\begin{algorithm}
    \caption{$\textsc{TSPviaMinPlusProduct}(n, c)$}
	\label{alg:new_dp}
\begin{algorithmic}[1]
    \State $\dprec(\vcentcolon, \vcentcolon) \gets \infty$
    \State $\dprec(\{1\}, 1) \gets 0$
    \For {$\ell = 2, \ldots, n$}
        \For {$\mathcal{B} \in \textsc{Batches}(\ell - 1)$}
            \State $p \gets \textsc{MinPlusProduct}(\dprec(\mathcal{B}, \vcentcolon), c)$
            \For {$i = 1, ..., |\mathcal{B}|$}
                \State $\dprec(\mathcal{B}[i] \cup \{k\}, k) \gets \min(\dprec(\mathcal{B}[i] \cup \{k\}, k), p_{i,k})$, $\forall k \in [n] \setminus \mathcal{B}[i]$
            \EndFor
        \EndFor
    \EndFor
    \State \Return $\displaystyle\min_{k}\left(\dprec([n], k) + c_{k,1}\right)$
\end{algorithmic}
\end{algorithm}

We show that the algorithm returns the optimal solution in the claimed running time:

\MainResult*
\begin{proof}
    The main observation is that Alg.~\ref{alg:new_dp} slightly changes the computation flow: While Eq.~\eqref{eq:rec} \emph{pulls} the values from the previous layer to compute the current layer, our algorithm \emph{pushes} them instead. Indeed, once the min-plus product of the current batch has been computed as matrix $p$, we can gradually update the values of the current layer $\ell$. Since all batches of the previous layer are iterated, the correctness of the previous algorithm is preserved.

    Let us now analyze the running time. For each layer $\ell$, there are $\lceil {n \choose \ell-1} \rceil$ batches to iterate. In each batch, we perform a min-plus product between the $\dprec$-table \emph{restricted} to the batch $\mathcal{B}$ (this has $n \times n$ size) and the $n \times n$ cost matrix in time $T(n)$. We can iterate the $n \times n$ intermediate matrix $p$ to update the values of the current layer. Thus, the total running time reads\[
       O\left(\sum_{\ell=2}^{n} \frac{\binom{n}{{\ell - 1}}}{n} \cdot \left(T(n) + n^2\right)\right) = O\left(\frac{2^n}{n} T(n)\right).
    \]
\end{proof}

Let us now take a closer look at $T(n)$, the running time of the min-plus matrix product. Fortunately, there has been a rich line of research on computing it, culminating in Williams' algorithm~\cite{williams_apsp}, which runs in time $n^3 / 2^{\Omega(\sqrt{\log n})}$ and was later derandomized by Chan and Williams~\cite{chan_williams}. Thus, our Alg.~\ref{alg:new_dp} can run in time $2^n n^2 / 2^{\Omega(\sqrt{\log n})}$.

\section{Related Work}

\sparagraph{Traveling Salesman Problem.}~The research on TSP is rather vast and it would take an entire survey to cover it. Apart from those mentioned in the introduction, approximation algorithms for TSP are also a beloved research area: Christofides' 1.5-approximation algorithm~\cite{christofides} (independently discovered by Serdyukov~\cite{serdyukov}) is probably one of the classics of approximation algorithms. In a breakthrough result, Karlin, Klein, and Gharan~\cite{karlin_tsp} showed that for some $\varepsilon > 10^{-36}$, there is an $(1.5-\varepsilon)$-approximation algorithm. Remarkable results have also been obtained for asymmetric TSP~\cite{asadpour2017log, ola_tsp} and in other settings as well~\cite{traub2019approaching, zenklusen20191}. Beyond that, there is also a PTAS for the Euclidean setting~\cite{arora_tsp, mitchell_tsp}. Another favorite research area is the $k$-OPT heuristic, which is a local search algorithm that is allowed to replace $k$ edges of a given solution to obtain a better cost~\cite{k_opt_1, k_opt_2, k_opt_3, k_opt_4}. A further prominent area is that of polynomial-space algorithms, under which setting TSP has an $O(4^n n^{\log n})$-time algorithm~\cite{poly_space_tsp_1, poly_space_tsp_2}. Alternatively, one can trade off time and space, as shown by Koivisto and Parviainen~\cite{tsp_tradeoff}. TSP also enjoys a quantum speedup, as shown by Ambainis et al.~\cite{quantum_speedup}.

\sparagraph{Min-Plus Matrix Product.}~While our improved algorithm for TSP directly uses the fastest deterministic algorithm for the min-plus product, which uses the polynomial method from circuit complexity~\cite{williams_apsp, chan_williams}, the line of research on this problem has been vast, with incremental improvements over the years, mainly for its main application, the all-pairs shortest paths (APSP) problem~\cite{apsp_1, apsp_2, apsp_3, apsp_4, apsp_5, apsp_6, apsp_7, apsp_8, apsp_9, apsp_10, apsp_11, chan_geo_apsp, apsp_13}. Notably, there is an open problem in fine-grained complexity whether there exists an algorithm for the min-plus product which runs in truly subcubic time, i.e., $O(n^{3-\varepsilon})$ for $\varepsilon > 0$~\cite{apsp_conjecture}. To this end, recent work has focused on understanding for which instances this product can be solved in truly subcubic time. Such instances are bounded-difference matrices~\cite{bounded_matrix_apsp_1, bounded_matrix_apsp_2}, later generalized to the case where one of the matrices is of $O(1)$-approximate rank~\cite{generalized_xu}, matrices with monotone rows or columns~\cite{monotone_apsp_1, monotone_apsp_3, monotone_apsp_2}, or geometrically weighted matrices~\cite{chan_geo_apsp}.


\section{Discussion}

While modest, our improvement shows that the well-known dynamic programming solution due to Bellman~\cite{bellman_dp} and, independently Held and Karp~\cite{held_karp_dp}, designed more than sixty years ago and taught in every algorithms course, is \emph{not} the best we can hope for. This breaks another \emph{psychological} barrier, as Koutis and Williams~\cite{koutis_williams} refer to it.

\sparagraph{Open Problems.}~A natural question is whether our algorithm can be sped up. Given the APSP hypothesis~\cite{apsp_conjecture}, it seems that we may indeed need to resort to special instances. For example, note that the cost matrix of the TSP instance remains constant throughout the algorithm. This opens up the possibility of preprocessing it, even in $O(2^n n)$-time, so that upcoming matrix products can be performed faster.

In this light, it is natural to ask whether an $O(2^n n^{2-\varepsilon})$-time algorithm exists for $\varepsilon > 0$. Even more interesting is the (still) open question of whether there is an $O(1.9999^n)$-time algorithm, for which the work of Nederlof~\cite{nederlof_tsp_bipartite} has already paved the way.

\sparagraph{Accessibility.}~For the ``TSP fans'' of Cook~\cite{cook_tsp}, we have prepared a proof-of-concept implementation of Alg.~\ref{alg:new_dp} on the author's Github page.\footnote{\url{github.com/stoianmihail}}


\bibliography{tsp}

\begin{thebibliography}{10}

\bibitem{quantum_speedup}
Andris Ambainis, Kaspars Balodis, J{\=a}nis Iraids, Martins Kokainis, Kri{\v{s}}j{\=a}nis Pr{\=u}sis, and Jevg{\=e}nijs Vihrovs.
\newblock Quantum speedups for exponential-time dynamic programming algorithms.
\newblock In {\em Proceedings of the Thirtieth Annual ACM-SIAM Symposium on Discrete Algorithms}, pages 1783--1793. SIAM, 2019.

\bibitem{arora_tsp}
Sanjeev Arora.
\newblock Polynomial time approximation schemes for euclidean traveling salesman and other geometric problems.
\newblock {\em Journal of the ACM (JACM)}, 45(5):753--782, 1998.

\bibitem{asadpour2017log}
Arash Asadpour, Michel~X Goemans, Aleksander M{\k{a}}dry, Shayan~Oveis Gharan, and Amin Saberi.
\newblock An o (log n/log log n)-approximation algorithm for the asymmetric traveling salesman problem.
\newblock {\em Operations Research}, 65(4):1043--1061, 2017.

\bibitem{bellman_dp}
Richard Bellman.
\newblock Dynamic programming treatment of the travelling salesman problem.
\newblock {\em Journal of the ACM (JACM)}, 9(1):61--63, 1962.

\bibitem{k_opt_2}
Mark~de Berg, Kevin Buchin, Bart~MP Jansen, and Gerhard Woeginger.
\newblock Fine-grained complexity analysis of two classic tsp variants.
\newblock {\em ACM Transactions on Algorithms (TALG)}, 17(1):1--29, 2020.

\bibitem{hamilton_determinant}
Andreas Bjorklund.
\newblock Determinant sums for undirected hamiltonicity.
\newblock {\em SIAM Journal on Computing}, 43(1):280--299, 2014.

\bibitem{poly_space_tsp_2}
Andreas Bj{\"o}rklund and Thore Husfeldt.
\newblock Exact algorithms for exact satisfiability and number of perfect matchings.
\newblock {\em Algorithmica}, 52:226--249, 2008.

\bibitem{tsp_bounded_1}
Andreas Bj{\"o}rklund, Thore Husfeldt, Petteri Kaski, and Mikko Koivisto.
\newblock The travelling salesman problem in bounded degree graphs.
\newblock In {\em International Colloquium on Automata, Languages, and Programming}, pages 198--209. Springer, 2008.

\bibitem{bounded_matrix_apsp_1}
Karl Bringmann, Fabrizio Grandoni, Barna Saha, and Virginia~Vassilevska Williams.
\newblock Truly subcubic algorithms for language edit distance and rna folding via fast bounded-difference min-plus product.
\newblock {\em SIAM Journal on Computing}, 48(2):481--512, 2019.

\bibitem{chan_geo_apsp}
Timothy~M Chan.
\newblock More algorithms for all-pairs shortest paths in weighted graphs.
\newblock In {\em Proceedings of the thirty-ninth annual ACM symposium on Theory of computing}, pages 590--598, 2007.

\bibitem{apsp_10}
Timothy~M Chan.
\newblock All-pairs shortest paths with real weights in o (n 3/log n) time.
\newblock {\em Algorithmica}, 50:236--243, 2008.

\bibitem{chan_williams}
Timothy~M Chan and Ryan Williams.
\newblock Deterministic apsp, orthogonal vectors, and more: Quickly derandomizing razborov-smolensky.
\newblock In {\em Proceedings of the twenty-seventh annual ACM-SIAM symposium on Discrete algorithms}, pages 1246--1255. SIAM, 2016.

\bibitem{bounded_matrix_apsp_2}
Shucheng Chi, Ran Duan, and Tianle Xie.
\newblock Faster algorithms for bounded-difference min-plus product.
\newblock In {\em Proceedings of the 2022 Annual ACM-SIAM Symposium on Discrete Algorithms (SODA)}, pages 1435--1447. SIAM, 2022.

\bibitem{monotone_apsp_1}
Shucheng Chi, Ran Duan, Tianle Xie, and Tianyi Zhang.
\newblock Faster min-plus product for monotone instances.
\newblock In {\em Proceedings of the 54th Annual ACM SIGACT Symposium on Theory of Computing}, pages 1529--1542, 2022.

\bibitem{christofides}
Nicos Christofides.
\newblock Worst-case analysis of a new heuristic for the travelling salesman problem.
\newblock In {\em Operations Research Forum}, volume~3, page~20. Springer, 2022.

\bibitem{cook_tsp}
William~J Cook.
\newblock {\em In pursuit of the traveling salesman: mathematics at the limits of computation}.
\newblock Princeton University Press, 2015.

\bibitem{k_opt_4}
Marek Cygan, {\L}ukasz Kowalik, and Arkadiusz Soca{\l}a.
\newblock Improving tsp tours using dynamic programming over tree decompositions.
\newblock {\em ACM Transactions on Algorithms (TALG)}, 15(4):1--19, 2019.

\bibitem{tsp_bounded_2}
Marek Cygan and Marcin Pilipczuk.
\newblock Faster exponential-time algorithms in graphs of bounded average degree.
\newblock {\em Information and Computation}, 243:75--85, 2015.

\bibitem{euclid_tsp_algo}
Mark De~Berg, Hans~L Bodlaender, S{\'a}ndor Kisfaludi-Bak, and Sudeshna Kolay.
\newblock An eth-tight exact algorithm for euclidean tsp.
\newblock In {\em 2018 IEEE 59th Annual Symposium on Foundations of Computer Science (FOCS)}, pages 450--461. IEEE, 2018.

\bibitem{euclid_tsp_lb}
Mark De~Berg, Hans~L Bodlaender, S{\'a}ndor Kisfaludi-Bak, D{\'a}niel Marx, and Tom~C Van Der~Zanden.
\newblock A framework for exponential-time-hypothesis--tight algorithms and lower bounds in geometric intersection graphs.
\newblock {\em SIAM Journal on Computing}, 49(6):1291--1331, 2020.

\bibitem{apsp_4}
Wlodzimierz Dobosiewicz.
\newblock A more efficient algorithm for the min-plus multiplication.
\newblock {\em International journal of computer mathematics}, 32(1-2):49--60, 1990.

\bibitem{monotone_apsp_2}
Anita D{\"u}rr.
\newblock Improved bounds for rectangular monotone min-plus product and applications.
\newblock {\em Information Processing Letters}, 181:106358, 2023.

\bibitem{apsp_1}
Robert~W Floyd.
\newblock Algorithm 97: shortest path.
\newblock {\em Communications of the ACM}, 5(6):345--345, 1962.

\bibitem{apsp_3}
Michael~L Fredman.
\newblock New bounds on the complexity of the shortest path problem.
\newblock {\em SIAM Journal on Computing}, 5(1):83--89, 1976.

\bibitem{monotone_apsp_3}
Yuzhou Gu, Adam Polak, Virginia~Vassilevska Williams, and Yinzhan Xu.
\newblock Faster monotone min-plus product, range mode, and single source replacement paths.
\newblock {\em arXiv preprint arXiv:2105.02806}, 2021.

\bibitem{poly_space_tsp_1}
Yuri Gurevich and Saharon Shelah.
\newblock Expected computation time for hamiltonian path problem.
\newblock {\em SIAM Journal on Computing}, 16(3):486--502, 1987.

\bibitem{apsp_6}
Yijie Han.
\newblock Improved algorithm for all pairs shortest paths.
\newblock {\em Information Processing Letters}, 91(5):245--250, 2004.

\bibitem{apsp_11}
Yijie Han.
\newblock An o (n 3 (log log n/log n) 5/4) time algorithm for all pairs shortest path.
\newblock {\em Algorithmica}, 51:428--434, 2008.

\bibitem{apsp_13}
Yijie Han and Tadao Takaoka.
\newblock An o(n3 loglogn/log2n) time algorithm for all pairs shortest paths.
\newblock {\em Journal of Discrete Algorithms}, 38-41:9--19, 2016.
\newblock URL: \url{https://www.sciencedirect.com/science/article/pii/S1570866716300296}, \href {https://doi.org/10.1016/j.jda.2016.09.001} {\path{doi:10.1016/j.jda.2016.09.001}}.

\bibitem{k_opt_3}
Sophia Heimann, Hung~P Hoang, and Stefan Hougardy.
\newblock The $k$-opt algorithm for the traveling salesman problem has exponential running time for $k \geq 5$.
\newblock {\em arXiv preprint arXiv:2402.07061}, 2024.

\bibitem{held_karp_dp}
Michael Held and Richard~M Karp.
\newblock A dynamic programming approach to sequencing problems.
\newblock {\em Journal of the Society for Industrial and Applied mathematics}, 10(1):196--210, 1962.

\bibitem{karlin_tsp}
Anna~R Karlin, Nathan Klein, and Shayan~Oveis Gharan.
\newblock A (slightly) improved approximation algorithm for metric tsp.
\newblock In {\em Proceedings of the 53rd Annual ACM SIGACT Symposium on Theory of Computing}, pages 32--45, 2021.

\bibitem{tsp_tradeoff}
Mikko Koivisto and Pekka Parviainen.
\newblock A space--time tradeoff for permutation problems.
\newblock In {\em Proceedings of the Twenty-First Annual ACM-SIAM Symposium on Discrete Algorithms}, pages 484--492. SIAM, 2010.

\bibitem{koutis_williams}
Ioannis Koutis and Ryan Williams.
\newblock Algebraic fingerprints for faster algorithms.
\newblock {\em Communications of the ACM}, 59(1):98--105, 2015.

\bibitem{k_opt_1}
Mark~W Krentel.
\newblock Structure in locally optimal solutions.
\newblock In {\em 30th Annual Symposium on Foundations of Computer Science}, pages 216--221. IEEE Computer Society, 1989.

\bibitem{mitchell_tsp}
Joseph~SB Mitchell.
\newblock Guillotine subdivisions approximate polygonal subdivisions: A simple polynomial-time approximation scheme for geometric tsp, k-mst, and related problems.
\newblock {\em SIAM Journal on computing}, 28(4):1298--1309, 1999.

\bibitem{nederlof_tsp_bipartite}
Jesper Nederlof.
\newblock Bipartite tsp in $o(1.9999^n)$ time, assuming quadratic time matrix multiplication.
\newblock In {\em Proceedings of the 52nd Annual ACM SIGACT Symposium on Theory of Computing}, pages 40--53, 2020.

\bibitem{schrijver_history}
Alexander Schrijver.
\newblock On the history of combinatorial optimization (till 1960).
\newblock {\em Handbooks in operations research and management science}, 12:1--68, 2005.

\bibitem{serdyukov}
Anatoliy~I Serdyukov.
\newblock On some extremal tours in graphs.
\newblock {\em Upravlyaemye systemy}, 17:76--79, 1978.

\bibitem{ola_tsp}
Ola Svensson, Jakub Tarnawski, and L\'{a}szl\'{o}~A. V\'{e}gh.
\newblock A constant-factor approximation algorithm for the asymmetric traveling salesman problem.
\newblock {\em J. ACM}, 67(6), nov 2020.
\newblock \href {https://doi.org/10.1145/3424306} {\path{doi:10.1145/3424306}}.

\bibitem{apsp_5}
Tadao Takaoka.
\newblock A new upper bound on the complexity of the all pairs shortest path problem.
\newblock {\em Information Processing Letters}, 43(4):195--199, 1992.

\bibitem{apsp_7}
Tadao Takaoka.
\newblock A faster algorithm for the all-pairs shortest path problem and its application.
\newblock In {\em Computing and Combinatorics: 10th Annual International Conference, COCOON 2004, Jeju Island, Korea, August 17-20, 2004. Proceedings 10}, pages 278--289. Springer, 2004.

\bibitem{apsp_9}
Tadao Takaoka.
\newblock An o (n3loglogn/logn) time algorithm for the all-pairs shortest path problem.
\newblock {\em Information Processing Letters}, 96(5):155--161, 2005.

\bibitem{traub2019approaching}
Vera Traub and Jens Vygen.
\newblock Approaching 3/2 for the s-t-path tsp.
\newblock {\em Journal of the ACM (JACM)}, 66(2):1--17, 2019.

\bibitem{apsp_2}
Stephen Warshall.
\newblock A theorem on boolean matrices.
\newblock {\em Journal of the ACM (JACM)}, 9(1):11--12, 1962.

\bibitem{williams_apsp}
Ryan Williams.
\newblock Faster all-pairs shortest paths via circuit complexity.
\newblock In {\em Proceedings of the forty-sixth annual ACM symposium on Theory of computing}, pages 664--673, 2014.

\bibitem{apsp_conjecture}
Virginia~Vassilevska Williams.
\newblock On some fine-grained questions in algorithms and complexity.
\newblock In {\em Proceedings of the international congress of mathematicians: Rio de janeiro 2018}, pages 3447--3487. World Scientific, 2018.

\bibitem{generalized_xu}
Virginia~Vassilevska Williams and Yinzhan Xu.
\newblock Truly subcubic min-plus product for less structured matrices, with applications.
\newblock In {\em Proceedings of the Fourteenth Annual ACM-SIAM Symposium on Discrete Algorithms}, pages 12--29. SIAM, 2020.

\bibitem{zenklusen20191}
Rico Zenklusen.
\newblock A 1.5-approximation for path tsp.
\newblock In {\em Proceedings of the thirtieth annual ACM-SIAM symposium on discrete algorithms}, pages 1539--1549. SIAM, 2019.

\bibitem{apsp_8}
Uri Zwick.
\newblock A slightly improved sub-cubic algorithm for the all pairs shortest paths problem with real edge lengths.
\newblock In {\em International Symposium on Algorithms and Computation}, pages 921--932. Springer, 2004.

\end{thebibliography}

\end{document}